\documentclass[conference]{IEEEtran}
\IEEEoverridecommandlockouts

\usepackage{cite}
\usepackage{amsmath,amssymb,amsfonts}
\usepackage{algorithmic}
\usepackage{graphicx}
\usepackage{textcomp}
\usepackage{xcolor}
\usepackage{subfigure}
\newtheorem{Proposition}{\textbf{Proposition}}
\newtheorem{proof}{\textbf{Proof}}
\def\BibTeX{{\rm B\kern-.05em{\sc i\kern-.025em b}\kern-.08em   T\kern-.1667em\lower.7ex\hbox{E}\kern-.125emX}}
\begin{document}

\title{Data-driven Method to Ensure Cascade Stability of Traffic Load Balancing in O-RAN Based Networks\\

\thanks{The work is supported by EPSRC CHEDDAR: Communications Hub For Empowering Distributed ClouD Computing Applications And Research (EP/X040518/1) (EP/Y037421/1).}
}

\author{\IEEEauthorblockN{1\textsuperscript{st} Mengbang Zou}
\IEEEauthorblockA{\textit{Engineering and Applied Science School} \\
\textit{Cranfield University, Bedford, UK}\\
m.zou@cranfield.ac.uk}
\and
\IEEEauthorblockN{2\textsuperscript{nd} Yun Tang}
\IEEEauthorblockA{\textit{Engineering and Applied Science School} \\
\textit{Cranfield University, Bedford, UK}\\
yun.tang@cranfield.ac.uk}
\and
\IEEEauthorblockN{3\textsuperscript{rd} Weisi Guo}
\IEEEauthorblockA{\textit{Engineering and Applied Science School} \\
\textit{Cranfield University, Bedford, UK}\\
weisi.guo@cranfield.ac.uk}
\and
}

\maketitle

\begin{abstract}
Load balancing in open radio access networks (O-RAN) is critical for ensuring efficient resource utilization, and the user's experience by evenly distributing network traffic load. Current research mainly focuses on designing load-balancing algorithms to allocate resources while overlooking the cascade stability of load balancing, which is critical to prevent endless handover. The main challenge to analyse the cascade stability lies in the difficulty of establishing an accurate mathematical model to describe the process of load balancing due to its nonlinearity and high-dimensionality. In our previous theoretical work, a simplified general dynamic function was used to analyze the stability. However, it is elusive whether this function is close to the reality of the load balance process. To solve this problem, 1) a data-driven method is proposed to identify the dynamic model of the load balancing process according to the real-time traffic load data collected from the radio units (RUs); 2) the stability condition of load balancing process is established for the identified dynamics model. Based on the identified dynamics model and the stability condition, the RAN Intelligent Controller (RIC) can control RUs to achieve a desired load-balancing state while ensuring cascade stability.
\end{abstract}

\begin{IEEEkeywords}
Open Radio Access Network (O-RAN), load balance, stability, complex network, system identification, data-driven
\end{IEEEkeywords}

\section{Introduction}
Load balancing is an important operation to homogeneously distribute the traffic demand across the network to improve the resource utilization efficiency, reduce energy consumption and ensure the Quality-of-Service (QoS) \cite{chang2022decentralized, shang2018wireless}. O-RAN \cite{o-ran_alliance}, as a disaggregated, open architecture, separates different network functions (Central Units (CUs), Distributed Units (DUs), and Radio Units (RUs)), providing opportunities for more flexible and efficient management of the radio access network (RAN) \cite{sroka2024policy}. In O-RAN-based networks, RUs serve different geographical areas or cell sectors. By balancing the traffic load across RUs, the network can make sure that resources like radio spectrum and power are efficiently utilized. Some research has studied the policy on how to allocate resources to realize load balancing in O-RAN. However, the cascade stability of load balancing in an O-RAN-based network is rarely studied, which is critical to prevent unnecessary handover in load balancing. 

\subsection{Literature in Load Balancing} 
Current research on load balancing mainly focuses on designing policies to select which base stations (BSs) to handover load from the overloaded BSs and the amount of load to hand over within the traditional radio access network (RAN) architectures. The load balancing problem can be resolved by tuning the logical BS boundaries automatically \cite{park2017mobility}. The traffic load of a BS is periodically monitored, and if it is overloaded, the BS shrinks its coverage area by modifying the parameter called cell individual offset (CIO) to hand over the edge users' service demand to the underloaded neighbouring BSs. The load balancing problem has been studied as an optimization problem by tuning the parameter CIO with different algorithms \cite{asghari2021reinforcement,alsuhli2021mobility, huang2022joint}. More recently, the load balancing problem within O-RAN architecture has been studied in \cite{zafar2024load, orhan2021connection, lai2023intelligent} by designing algorithms e.g. deep reinforcement learning methods to allocate user demands to different RUs. 

The stability of the ideal load balancing state means that all RUs can maintain the ideal balancing state even when perturbations happen, e.g. as users come and go. A stable load balancing mechanism ensures the efficient utilization of resources and avoids frequent user equipment (UE) handovers between RUs (the Ping-Pong effect \cite{zidic2023analyses}). For an arbitrarily large network, instability can cause cascade handover (e.g. user handover in RU-A causes unintended handovers in RU-B, RU-C, etc.), resulting in energy consumption problems without improving QoS. Our previous works \cite{10901638, moutsinas2019probabilistic} have established the stability condition for load balancing and the relationship between load balancing dynamics and network topology. However, they are based on the assumption that there exists a continuous function that adequately describes the load-balancing process. Since it is difficult to establish an accurate mathematics model to describe the dynamic process of load balancing because of the nonlinearity, high dimensionality and heterogeneity (e.g. the introduction of sleep mode \cite{wu2012traffic}, heterogeneous deployment (deploy macro, micro, and pico cells) \cite{wang2023base}, as well as other policies to reduce energy consumption), our previous theory is based on a simplified general dynamic equation instead of a specific one. Therefore, it is not clear whether the assumed dynamic equation is close to the reality of the O-RAN-based network, which makes it difficult to apply in an O-RAN network without an exact dynamics model.
In O-RAN, RUs report real-time traffic load data to the CU and the RAN Intelligent Controller (RIC). Based on the policies and commands received from the RIC, the DUs execute decisions on handing over user demands between RUs. The collected time-series traffic load data from RUs thus reflect the underlying load balancing dynamics.
Therefore, with enough time-series data, we can identify the load dynamics of the system and revise the policy based on the underlying dynamics model to ensure cascade stability in the traffic load-balancing process. 

\subsection{Novelty \& Contribution} 

The novelty and contributions of this paper are summarized as follows:

\begin{itemize} 
    \item \textbf{Data-Driven Dynamic Model Identification:}
    We propose a novel data-driven methodology to accurately identify the load balancing dynamics from real-time traffic load data collected by RUs in an O-RAN network. This approach addresses the critical gap in existing methods that often rely on oversimplified mathematical models.

    \item \textbf{Cascade Stability Analysis}:  
    Based on the identified dynamics, we analyze rigorous stability conditions that ensure cascade stability of the load balancing mechanism. The cascade stability condition is valid in any graph structure and does not require RUs have identical dynamics. This contribution is significant in preventing cascading handovers and associated energy wastage and QoS degradation.

    \item \textbf{Practical RIC Integration:} We blueprint a practical implementation architecture in RIC for real-time traffic load monitoring, dynamics identification, and load balancing policy deployment.

\end{itemize}

\begin{figure}[t]
    \centering
    \includegraphics[width=\columnwidth]{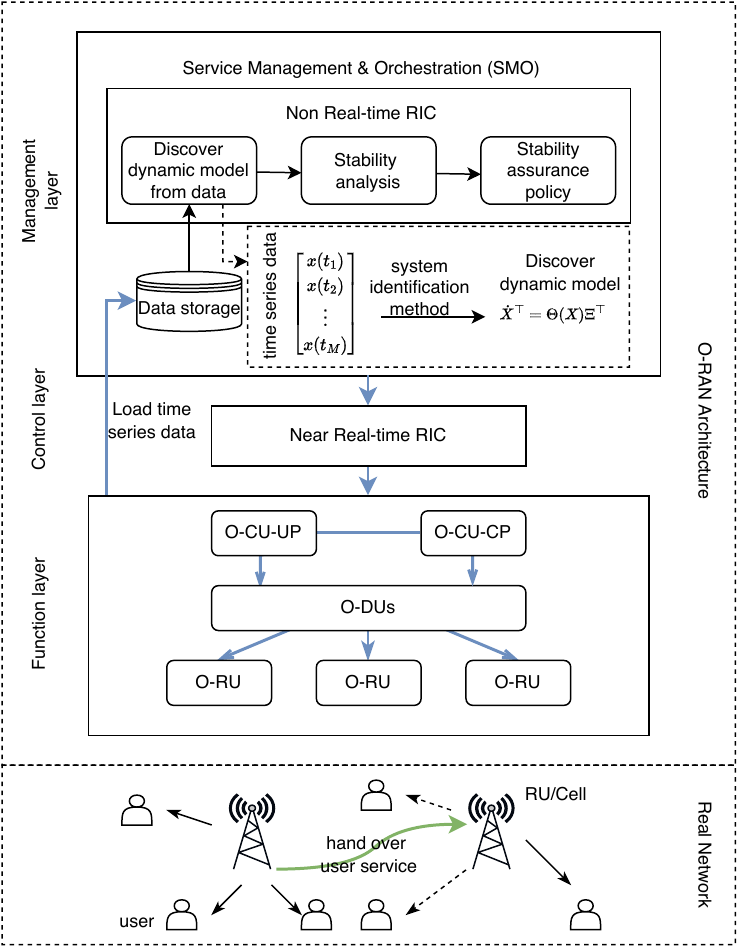}
    \caption{The implementation architecture in O-RAN network. RUs report traffic load to CU via DU through open interface. CU forward load information to Non-Real-time RIC to identify dynamics of RUs and analyze the stability. Near Real-time RIC makes decisions to distribute load among RUs.}
    \label{fig:O-RAN}
\end{figure}

\section{Methodology}
\subsection{System Setup}
The architecture of O-RAN is illustrated in  Fig. \ref{fig:O-RAN}.
In O-RAN, network functions are distributed across O-RAN such as service management and orchestration (SMO), Non Real-time RIC, Near Real-time RIC, CU, DU and RU. These functions communicate with each other through the open interface. In this paper, we consider $1$ CU, $1$ DU and $N$ RUs.

\begin{enumerate}
    \item RU measures traffic load (e.g., PRB utilization) and sends the data to DU via the Open Fronthaul (FH) interface. The DU processes RU reports and sends load data to the CU over the F1 interface. 

    \item The CU forwards load information to Non-Real-time RIC to train a model to identify the dynamics of load balancing and analyze the stability condition. 

    \item Near Real-time RIC makes real-time decisions to distribute load according to the policy configurations from the Non-Real-time RIC. 
\end{enumerate}

\subsection{Model \& Assumptions}
The traffic load definition in this paper is related to the physical resource block (PRB).
For user $u$, the maximum transmission rate over one PRB is 
\begin{equation}
    R_u^t=\alpha\log_2(1+\frac{P_i(t)G_{u,i}(t)}{N_0+\sum_{j \ne i}P_j(t)G_{u,j}(t)}),
\end{equation}
where $\alpha$ is the spectrum bandwidth of one PRB, $P_i(t)$ is the transmission power of RU $i$ at time $t$, $G_{u,i}$ is the channel power gain between user $u$ and RU $i$ at time $t$, $N_0$ is the noise power. Assume that each user has a constant bit rate requirement $M_u(t)$ at time $t$. The number of required PRBs to meet the demand $M_u(t)$ is given by
\begin{equation}
    B_u(t) = {\rm min}\{\frac{M_u^t}{R_u^t}, B_c\},
\end{equation}
where $B_c$ is a constant threshold to limit the number of PRBs of users with poor signal quality to be under a reasonable level \cite{xu2019deep}. The load of RU $i$ is 
\begin{equation}
    l_i(t) = \frac{\sum_{u=1}^{U_i(t)}\tilde{B}_{u}(t)}{B_i(t)},
\end{equation}
where $B_i(t)$ denotes the total number of PRBs in RU $i$. Without loss of generality, given load measure data, the load dynamics in RU $i$ can be described by a continuous function with $t$ as $\frac{dl_i}{dt}=f_i(l_i(t))$.

The handover process aims to transfer the user's service from its serving RU to a neighbouring RU. The handover process is triggered according to event A3, defined by 3GPP, 
\begin{equation}\label{equ: handover}
    M_j+\theta_{j\to i}>Hys + M_i + \theta_{i \to j},
\end{equation}
where $M_i$ and $M_j$ are the user-measured values of reference signal received power (RSRP) from BS $i$ and $j$, $\theta_{i \to j}$ is the BS individual offset value of BS $i$ concerning BS $j$. $Hys$ is a hysteresis parameter defined for all RUs. The CIO is an offset that can be applied to alter the handover decision, which affects the radio service coverage of RU. Traffic load can be handed over to neighbour RUs with relatively low traffic load by adjusting the value of CIO, if the traffic load is high in the RU. The load balancing policy is based on the relative load levels compared to the neighbouring RUs. Therefore, the offload dynamics between two nodes should be a function in terms of the load $l_i$ and $l_j$, i.e. $g(l_i, l_j)$. 

\subsection{Data driven method to identify the dynamics}
The load dynamics of each RU in a wireless network within the O-RAN is 
\begin{equation}\label{equ: coupling_dynamics}
    \dot{l}_i = f_i(l_i)+\sum_{j=1, j\ne i}^Na_{ij}g_{ij}(l_i, l_j),
\end{equation}
where $f_i(l_i)$ is the self-load dynamics of a RU related to the PRB, $g_{ij}(l_i, l_j)$ is the offloading dynamics between RUs related to the handover process as discussed before. $a_{ij}$ is the element of the connection matrix of the network. If $a_{ij}=1$, then RU $i$ and $j$ can share the load with each other. Otherwise, $a_{ij}=0$, the handover can not happen between them. We only know that the load dynamics function of each RU is decided by its load and the neighbouring RUs' load. The exact equation expression of the dynamics is still unknown. 

It is challenging to establish the mathematics models for a high-dimensional non-linear system in the real world, e.g. the accurate dynamics equations to describe the load dynamics and offloading dynamics. However, we can obtain abundant time-series load data from a real network, e.g. BubbleRAN \cite{bimo2023design}. This enables us to derive a data-driven system identification method to fit this dynamics process based on the data from the real system. 

First, time-series data is collected from the O-RAN equipment (BubbleRAN) or the simulator and formed into a data matrix: $\mathbf{L}=[\mathbf{l}(t_1) \quad \mathbf{l}(t_2) \cdots \mathbf{l}(t_m)]^{\top}$. We need to get the load rate matrix $\dot{\mathbf{L}}=[\dot{l}(t_1) \quad \dot{l}(t_2) \cdots \dot{l}(t_m)]^{\top}$ which can be estimated by $\dot{l}_i(t) \approx \frac{l(t+\Delta t)-l(t)}{\Delta t}$. Let $\Theta_f(l_i)$ be the library of candidate non-linear functions for self dynamics $f_i$ of RU $i$, e.g. the polynomials $\Theta_f(l_i)=[1 \quad l_i \quad l_i^2 \cdots l_i^d]$. The offloading dynamics are based on the relative load levels $l_i-l_j$, so it is natural to consider the library of candidate functions as $\Theta_g(l_i, l_j) = [a_{ij}(l_i-l_j) \quad a_{ij}(l_i-l_j)^2 \cdots a_{ij}(l_i-l_j)^d]$. Then the integrated library for RU $i$ is $\Theta(l_i) = [\Theta_f(l_i) \quad \Theta_g(l_i, l_j)]$. 

The dynamics of each RU can be represented in terms of the data and the function library as 
\begin{equation}
    \dot{l_i} = \Theta(l_i)\xi_i,
\end{equation}
where $\xi_i$ is the coefficient vector, which can be identified using a sparse regression as 
\begin{equation} \label{equ: sindy}
  \xi_i = \mathop{\arg\min}\limits_{\xi_i'}||\dot{l}_i-\Theta(l_i)\xi_i'||_2+ \gamma||\xi_i'||_1,
\end{equation}
where $\gamma$ is the parameter to control sparsity. The exact load dynamics of each RU can be obtained by solving equation~(\ref{equ: sindy}).

\subsection{Assure stability based on the identified dynamics}
Based on the load dynamics function, we can analyze the stability of the load balancing process. Now rewrite the load dynamics based on the function library with the connection matrix as 
\begin{equation}\label{equ: identify}
    \dot{l}_i = \Theta_f(l_i)\xi_{i}^f+ 
    \Theta_g(l_i, l_j)\xi_i^g
\end{equation}
The desired state for the service efficiency of the load balancing state is $l_1 = l_2 =\cdots l_N=1$. To ensure the stability of the desired state, we need to clarify the stability condition for the system.

\begin{Proposition}
    If the self-load dynamics of each RU satisfies that $f_i(l_i=1)=0, f_i'(l_i)<0$, the offloading dynamics rate between neighbours $i$ and $j$ at $l_i=1$ is the same, $\frac{\partial g_{ij}(l_i,l_j)}{\partial l_i}=\frac{\partial g_{ji}(l_j,l_i)}{\partial l_j}<0$, then the load balancing process can assure stability at the desired state $l_1=l_2=\cdots l_N=1$.
\end{Proposition}

\begin{proof}
    The dynamics of RU $i$ around $l_i=1$ can be estimated by the first order Taylor expansion around $l_i=1$ as 
    \[
        \Delta \dot{l}_i = f'_i(l_i)\Delta l_i+\sum a_{ij}(\frac{\partial g_{ij}(l_i, l_j)}{\partial l_i}\Delta l_i+\frac{\partial g_{ij}(l_i, l_j)}{\partial l_j}\Delta l_j)
    \]
    According to the data-driven method to identify the system, $g_{ij}(l_i, l_j)$ can be represented by the library of candidate functions, i.e. the polynomials $[1 \quad l_i-l_j \quad (l_i-l_j)^2 \cdots (l_i-l_j)^d]$.  $g_{ij}(l_i,l_j)=\sum_{k=0}^{d}p_{ij}(l_i-l_j)^k$. $\frac{\partial g_{ij}(l_i, l_j)}{\partial l_i}=p_{ij}$, $\frac{\partial g_{ij}(l_i, l_j)}{\partial l_j}=-p_{ij}$ at $l_i =l_j =1$. Write $\Delta \dot{l}_i$ in a matrix form $\Delta\dot{\mathbf{l}} = \mathbf{F}'(1)\Delta \mathbf{l}+ (\mathbf{Q}-\mathbf{A} \odot \mathbf{P})\Delta \mathbf{l}$, where $\mathbf{F}'(1)$ is a diagonal matrix $\mathbf{F}'(1)= {\rm diag}\{f_1'(1),f_2'(1),\cdots, f_N'(1)\}$. $\odot$ is a Hadamard product, $\mathbf{A}$ is the connection matrix, $\mathbf{P}$ is the matrix with the element $p_{ij}$. $\mathbf{Q}$ is a diagonal matrix with element $q_{ii} = \sum_j a_{ij}p_{ij}$. If $\frac{\partial g_{ij}(l_i,l_j)}{\partial l_i}=\frac{\partial g_{ji}(l_j,l_i)}{\partial l_j}$, then $p_{ij}=p_{ji}$, $\mathbf{P}$ is a symmetric matrix. $\mathbf{A}$ is a symmetric matrix, then $\mathbf{A}\odot\mathbf{P}$ is a symmetric matrix and $\mathbf{Q}-\mathbf{A}\odot\mathbf{P}$ is a Laplacian matrix. Let's assume that $\mathbf{K}= \mathbf{F}'(1)+\mathbf{Q}-\mathbf{A}\odot\mathbf{P}$. According to the Gershgorin circle theorem \cite{varga2011gervsgorin},

\begin{equation}\label{equ: gershgorin}
    |\lambda_i - K_{ii}| \le \sum_{j=1, j\ne i}|K_{ij}|.
\end{equation}
Since $|K_{ii}|>\sum_{j=1, j\ne i}|K_{ij}|$, $K_{ii}<0$, $\lambda_i < 0$. $\Delta \dot{\mathbf{l}}_i = \mathbf{K}\Delta \mathbf{l}$, with all eigenvalues of $\mathbf{K}$ are negative when $l_1 = l_2 = \cdots = l_N =1$. Therefore, the equilibrium is stable.
\end{proof} 

Therefore, to make the system maintain stability at the desired state the identified dynamic function in equation (\ref{equ: identify}) should satisfy Proposition 1. RIC identify the dynamics of each RU from their reported data and then controls the RUs' load balancing process to satisfy the stability condition.

\section{Simulation \& Results} 
Extensive simulations were conducted using an in-house O-RAN load simulator to validate the proposed method under realistic network conditions. 12 Radio Units (RUs) were uniformly distributed within a rectangular area (though the method is also applicable to stochastic geometry models), serving randomly generated User Equipments (UEs). The RUs featured adjustable Physical Resource Block (PRB) resources, coverage areas, and customizable policies for handover and PRB allocation, closely mimicking real-world dynamics.

\subsection{Setup and Scenario} 
In the simulation, UEs were randomly spawned with configurable counts, PRB demands, and lifetimes. Each UE initially connected to its geographically nearest RU, triggering handovers to neighbouring RUs when resource constraints arose. RUs managed PRB resources dynamically, adjusting between minimum thresholds and maximum capacities based on current load conditions. UEs received resources on a first-come-first-served basis, up to their requested PRBs, limited by the RU's available capacity. When an RU reached full PRB utilization and could not fully serve all connected UEs, handovers occurred to redistribute underserved UEs to neighbouring RUs with available PRB resources.

\begin{figure}[ht]
    \centering
    \renewcommand{\thesubfigure}{}
    \subfigure[]{
 \includegraphics[width=8cm, height=4cm]{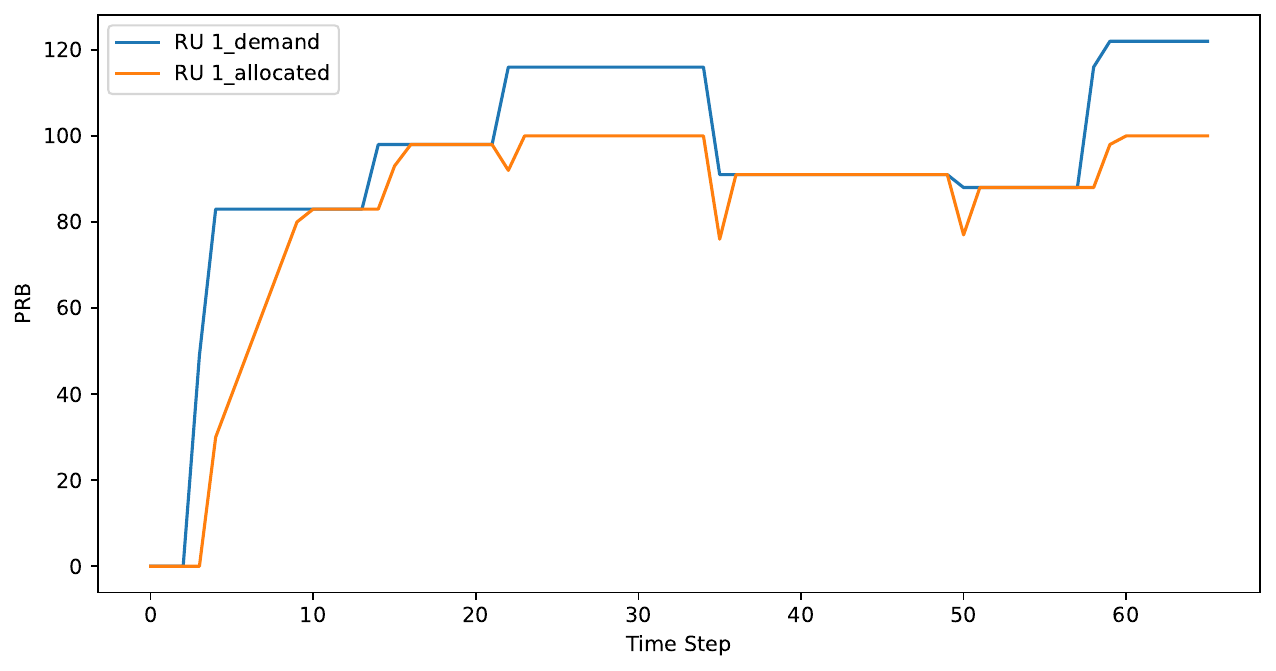}}
    
    \subfigure[]{
\includegraphics[width=8cm, height=4cm]{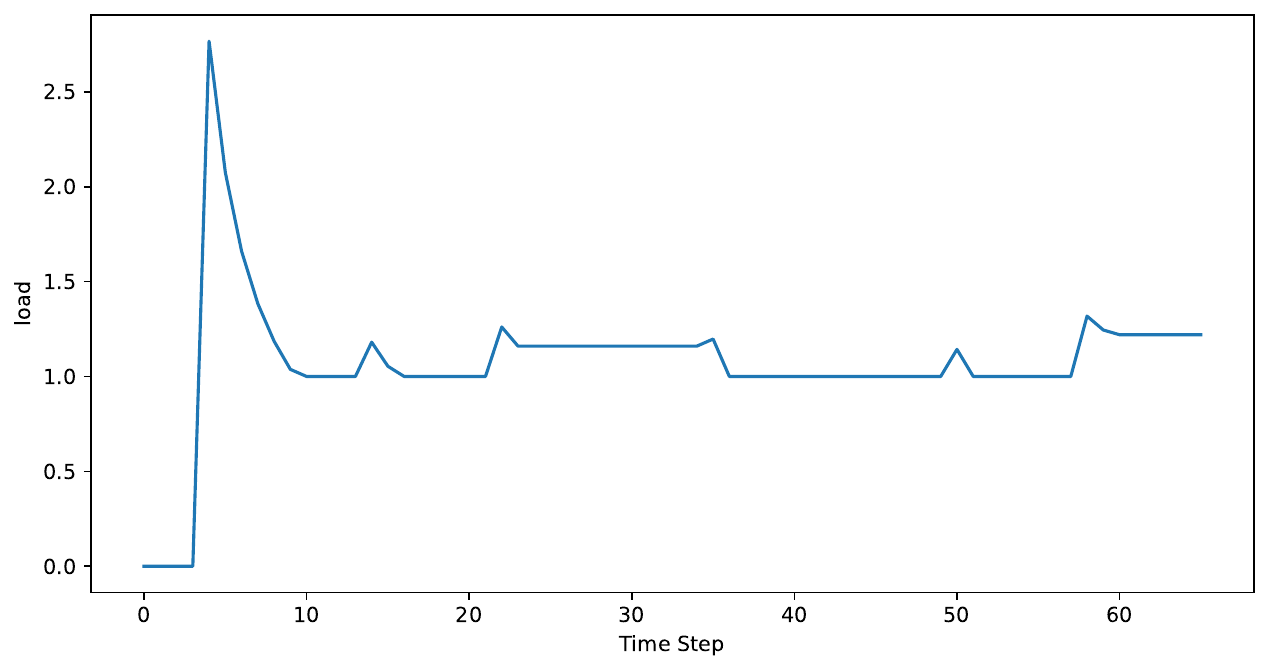}}
    
    \caption{(a) The time series data of demand PRB and allocated demand of RU 1. (b) The time series data of RU 1 load}
    \label{fig: load_time}
\end{figure} 

\begin{figure}[t]
    \centering
\renewcommand{\thesubfigure}{}
\includegraphics[width=\columnwidth]{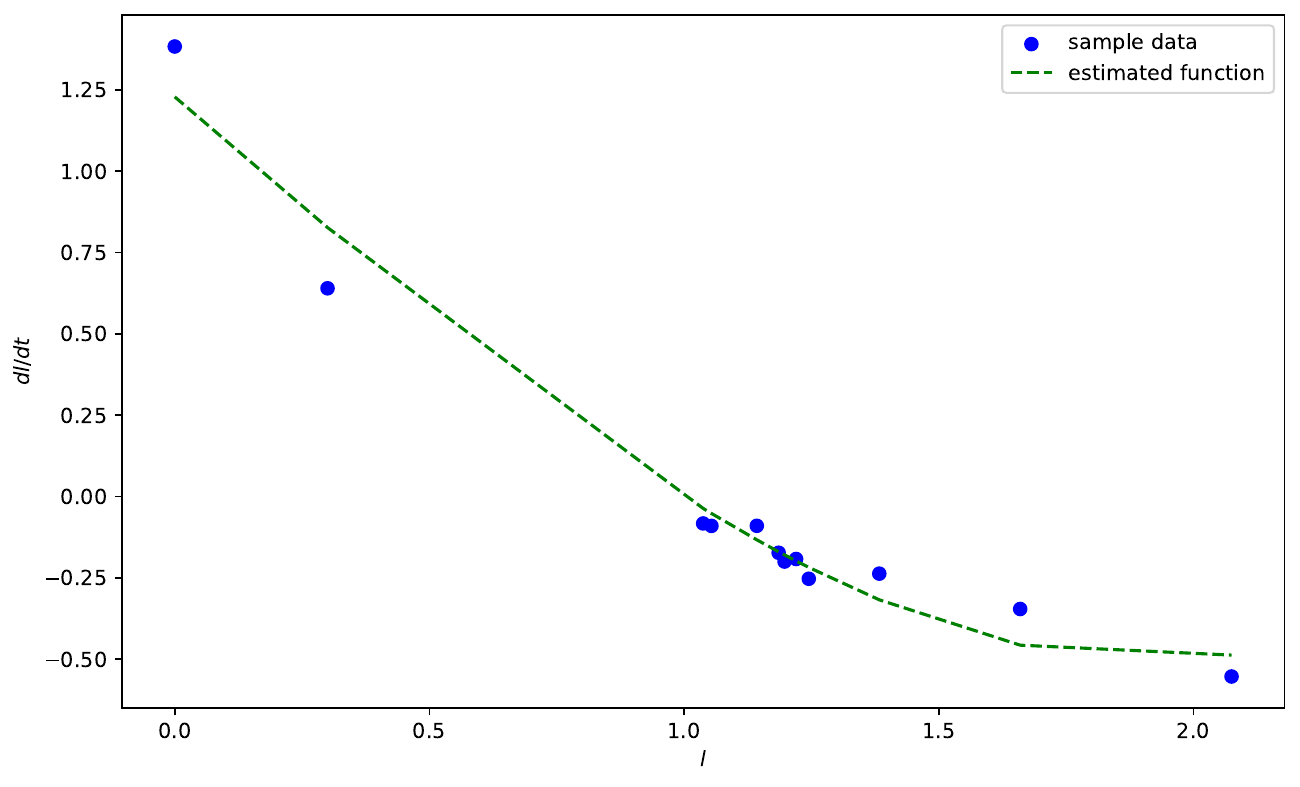}
    \caption{Identify the self-load dynamics from the collected data of RU 1}
    \label{fig: sindy_self}
\end{figure}

\begin{figure}[ht]
    \centering
\renewcommand{\thesubfigure}{} 
    \subfigure[]{
\includegraphics[width=4.2cm, height=3cm]{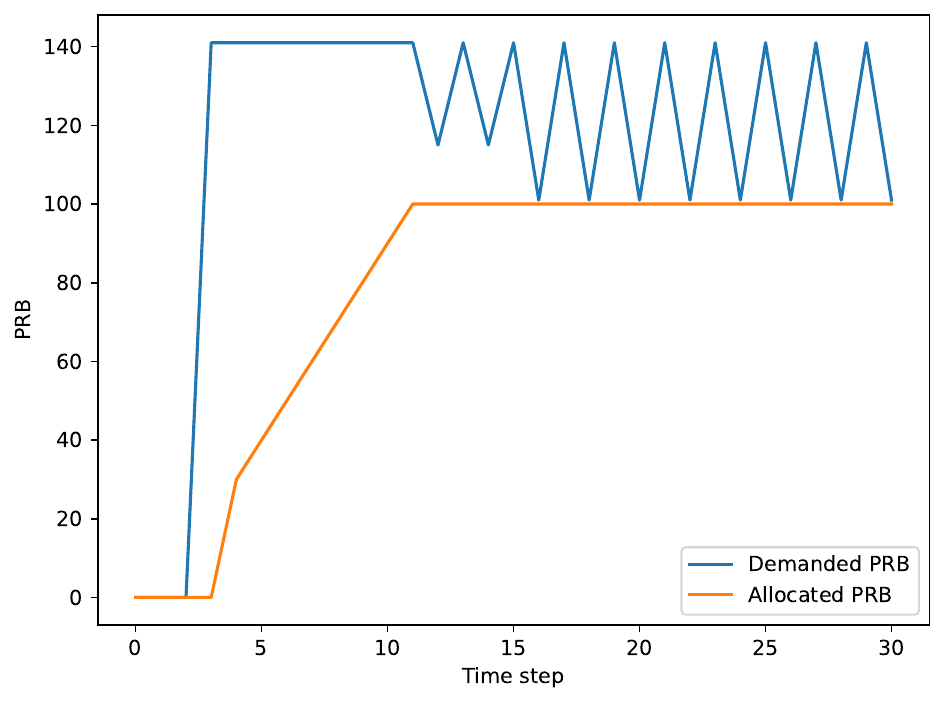}}
    \subfigure[]{\includegraphics[width=4.2cm, height=3cm]{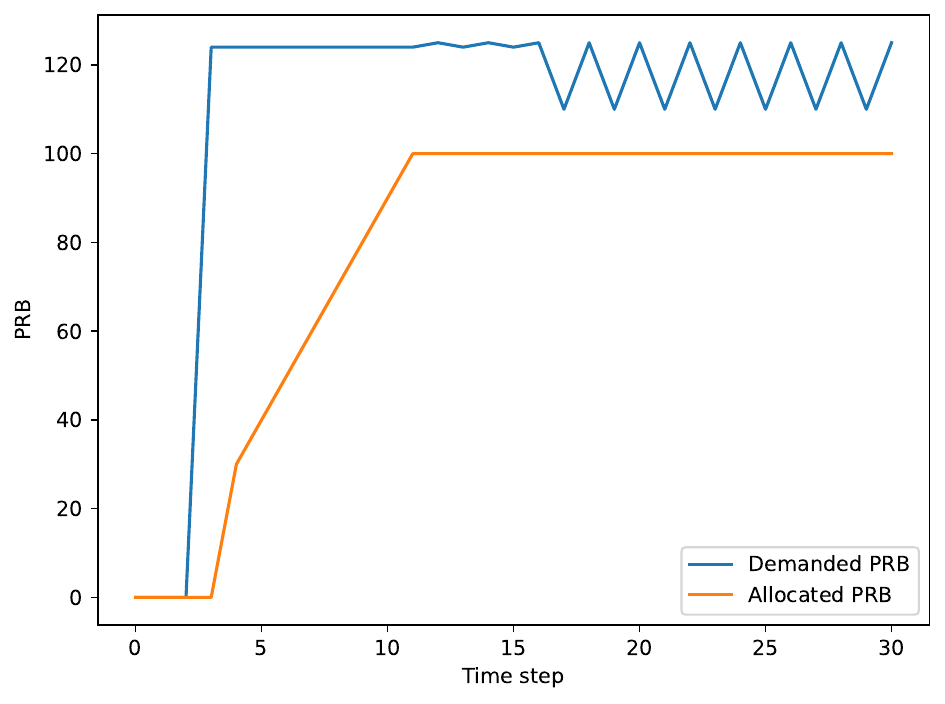}}
    
    \subfigure[]{\includegraphics[width=4.2cm, height=3cm]{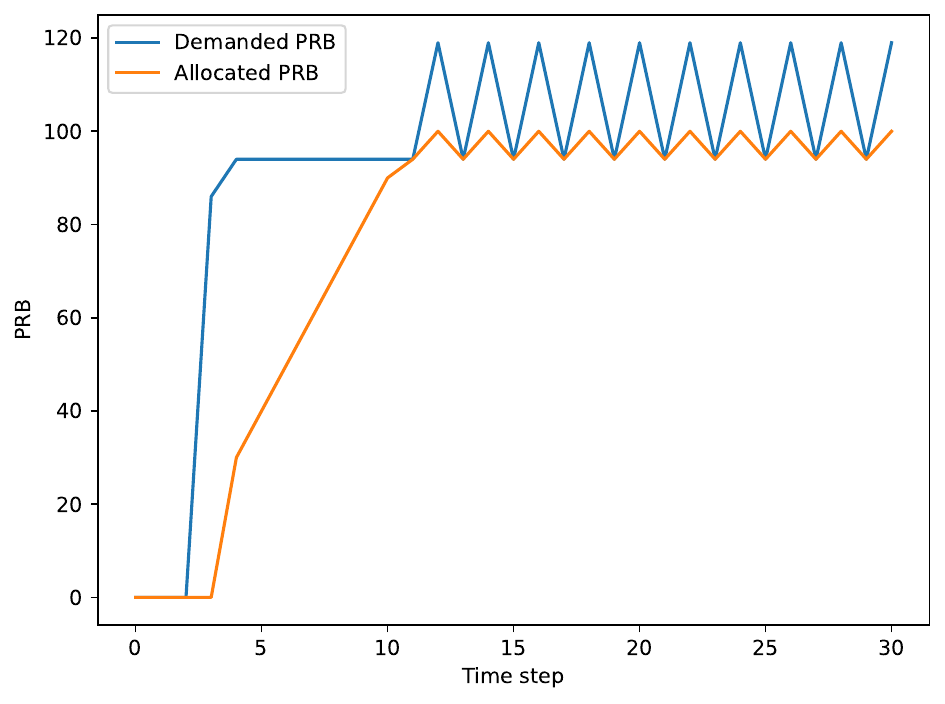}}
    \subfigure[]{\includegraphics[width=4.2cm, height=3cm]{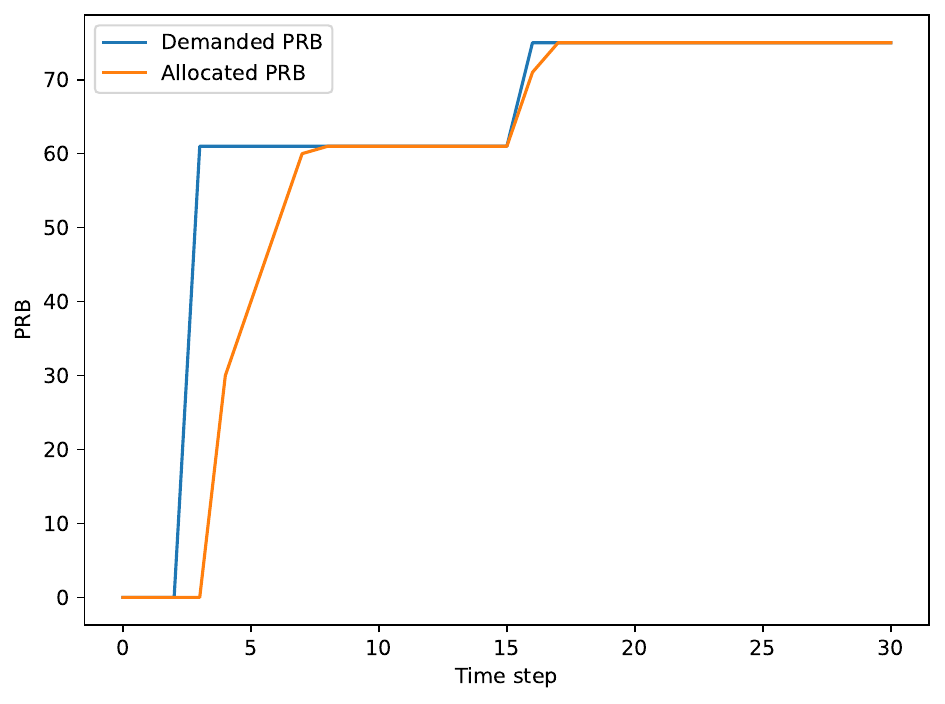}}

    \caption{Ping-pong effect among three RUs.}
    \label{fig: pingpong}
\end{figure} 

\begin{figure}[t]
    \centering
\renewcommand{\thesubfigure}{}
\includegraphics[width=\columnwidth]{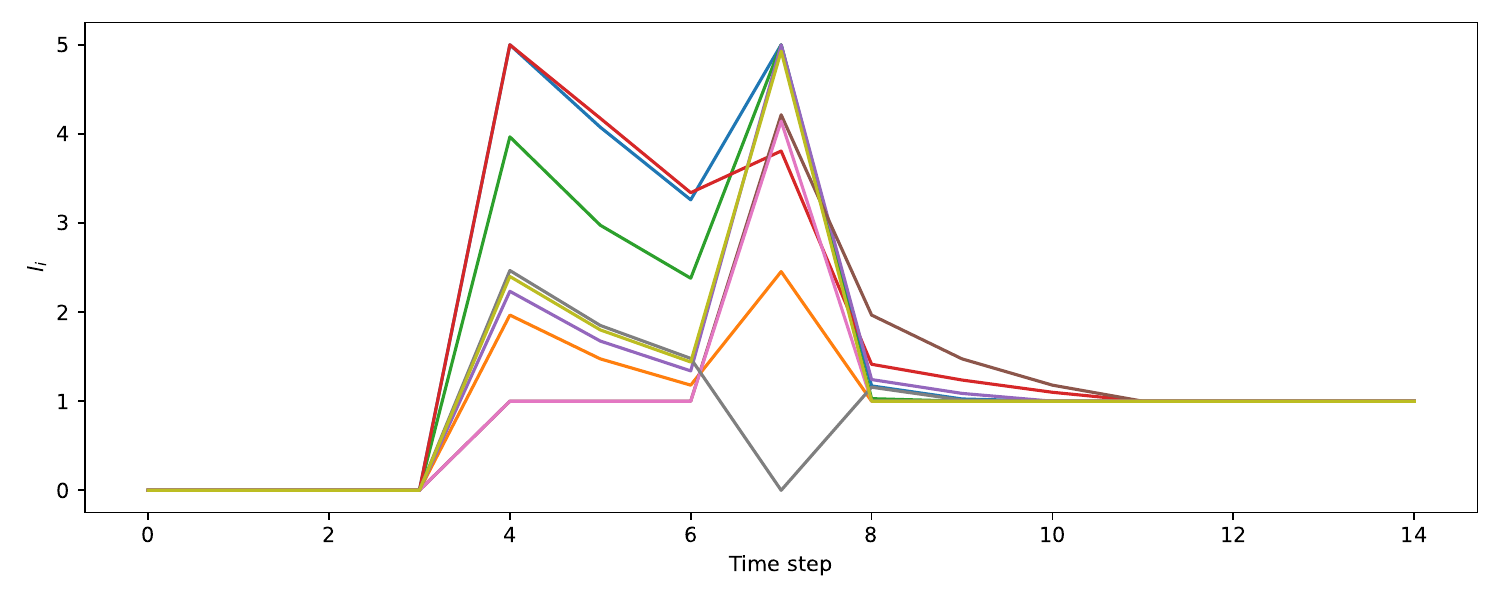}
    \caption{RUs synchronize to the desired state and maintain stability.}
    \label{fig: stable}
\end{figure}
\subsection{Results and Observations}
First, we collect the time series data of RUs without handover. The dynamics of RU 1 are shown in Fig. \ref{fig: load_time}. We can see that when users' service demand increases, RU adaptively adjusts the allocated PRB to satisfy users' service demand. Correspondingly, the load of RU 1 starts to decrease and achieve the desired state $l_i=1$. At around  ${\rm time step}=20$, the demand increases while the allocated PRB decreases, this is because during this time some users left the coverage area of RU 1 while other new users joined this area. So the RU initially decreased the allocated PRB and then increased PRB to satisfy users' demand. This is common in reality since it takes some time for RU to respond to the demand variation. However, when the service demand of users increased over 100 PRB, which is larger than the maximum capability of RU 1, RU 1 cannot satisfy all service demand. Without the handover of the service, the load of RU 1 can not return to the desired state. With the load time series data of RUs, we can identify the self-dynamics of each RU by the proposed data-driven method in Fig. \ref{fig: sindy_self}, $\dot{l}_1 = 1.229-1.35 l_1 + 0.122 l_1^3$. The self dynamics of RU satisfies the condition $\frac{d\dot{l}_i}{dl_i}<0$ at $l_i=1$, so $l_i$ is stable. 

Now consider the handover situation. 12 RUs can offload to their neighbours. When the handover policy does not satisfy the proposition that $\frac{\partial g_{ij}(l_i,l_j)}{\partial l_i}=\frac{\partial g_{ji}(l_j,l_i)}{\partial l_j}<0$. The introduction of new users' demand causes the handover effect among three RUs, even another RU can stabilise at the desired state (shown in Fig. \ref{fig: pingpong}). In this case, the RIC detects this abnormal situation, and adaptively adjusts the policy to make the load balancing process satisfies the stable condition. The dynamics can be identified by the proposed data-driven method. As a result, RUs synchronize to the desired stable state (shown in Fig. \ref{fig: stable}).

These results confirm that the proposed data-driven dynamic modelling and stability analysis approach effectively assures O-RAN load balancing stability and user experience.





\section{Conclusion}
In this paper, we introduced a novel data-driven approach to accurately model load balancing dynamics and ensure cascade stability in O-RAN networks. The primary contributions include a practical method for dynamic model identification from RU traffic data, a theoretical framework for stability analysis based on the identified dynamics, and integration with RIC for real-time control to ensure the cascade stability of the load balancing process. Simulation results demonstrate substantial improvements in cascade stability, reduced Ping-Pong handover effects, and balanced load distribution compared to the load balancing policy without stability assurance. By leveraging real-time network data and stability-aware control strategies, this work significantly contributes to enhancing O-RAN reliability, efficiency, and QoS.
Future research directions include investigating adaptive and robust modelling approaches for real-time deployment and exploring more robust methodology that deals with noisy data.

\bibliographystyle{IEEEtran}
\bibliography{ref}

\end{document}